\documentclass[conference]{IEEEtran}
\IEEEoverridecommandlockouts
\usepackage{cite}
\usepackage{amsmath,amssymb,amsfonts}
\usepackage{algorithmic}
\usepackage{graphicx}
\usepackage{textcomp}
\usepackage{xcolor}
\usepackage{amsthm}
\usepackage{booktabs}   
\usepackage{multirow}   
\usepackage{array}      
\usepackage{siunitx}    

\makeatletter
\def\@IEEEtablestring#1{\footnotesize #1}
\makeatother

\theoremstyle{plain}
\newtheorem{theorem}{Theorem}
\newtheorem{lemma}{Lemma}
\newtheorem{proposition}{Proposition}
\newtheorem{corollary}{Corollary}

\theoremstyle{definition}
\newtheorem{assumption}{Assumption}

\theoremstyle{remark}

\def\BibTeX{{\rm B\kern-.05em{\sc i\kern-.025em b}\kern-.08em
    T\kern-.1667em\lower.7ex\hbox{E}\kern-.125emX}}
\begin{document}

\title{Credit Limits beyond Full Collateralization in Decentralized Micropayments: Incentive Conditions}

\author{
\IEEEauthorblockN{Chien-Chih Chen}
\IEEEauthorblockA{\textit{Electrical and Computer Engineering} \\
\textit{University of Waterloo} \\
Waterloo, ON, Canada \\
j2255che@uwaterloo.ca}
\and
\IEEEauthorblockN{Wojciech Golab}
\IEEEauthorblockA{\textit{Electrical and Computer Engineering} \\
\textit{University of Waterloo} \\
Waterloo, ON, Canada \\
wgolab@uwaterloo.ca}
}

\maketitle

\begin{abstract}
In decentralized non-custodial micropayments, the central challenge is not whether payments can be executed directly, but under what conditions such systems can offer credit limits without requiring full collateral backing. Existing approaches typically tie available credit to posted collateral, causing liquidity requirements to scale with transaction volume and settlement exposure and limiting the practical usefulness of credit-based micropayments. This paper characterizes the incentive conditions under which credit-based non-custodial micropayments can operate beyond full collateralization while remaining incentive compatible. We model repeated buyer--merchant interactions under public monitoring and identify the roles of bounded exposure, verifiable settlement outcomes, and continuation value in deterring strategic default under non-custodial execution. The resulting characterization clarifies the trade-off between capital efficiency and the enforcement conditions required to sustain under-collateralized credit expansion without custodial trust. As an illustrative application-layer instantiation, an Arbitrum Nitro prototype provides execution-level evidence that the settlement, commitment, and incentive-enforcement paths of a credit-limit-based design can be realized with low on-chain overhead.
\end{abstract}


\section{Introduction}
\label{sec_micropayment_intro}

In decentralized non-custodial micropayments, the central challenge is not whether payments can be executed directly, but under what conditions such systems can offer credit limits without requiring full collateral backing. In conventional non-custodial designs, available credit is tightly coupled to posted collateral, causing liquidity requirements to scale with transaction volume and settlement exposure. This coupling induces capital lock-in and limits the practical usefulness of micropayment systems in settings where repeated low-value transactions would benefit from short-term credit.

This paper studies the incentive conditions under which decentralized non-custodial micropayment systems can support credit-based operation beyond full collateralization. Our focus is not on a single protocol instance, but on a broader design question: when can a non-custodial micropayment system provide credit limits while preserving incentive compatibility under strategic buyer--merchant interaction? Framed this way, the problem is fundamentally one of incentive design under capital constraints. Any mechanism that aims to support under-collateralized credit in a non-custodial environment must address the trade-off between capital efficiency and the enforcement conditions required to deter strategic default.

We characterize this trade-off through a repeated-game model of buyer--merchant interaction under public monitoring, bounded exposure, and asymmetric incentives. Within this framework, we identify conditions under which credit-based non-custodial micropayments can sustain conforming behavior beyond fully collateralized operation. On the buyer side, credit expansion must be disciplined by continuation value and bounded default incentives. On the merchant side, timely fulfillment must remain incentive compatible under publicly verifiable settlement outcomes and protocol-enforced penalties. The resulting characterization identifies a regime in which under-collateralized credit can be supported without custodial control, while also clarifying the assumptions that such support requires.

To connect the characterization to an execution setting, we use an application-layer instantiation with evolving credit limits and controlled over-limit liquidity provision. This instantiation serves as an illustrative implementation substrate rather than the main contribution of the paper. A prototype on Arbitrum Nitro is used only to demonstrate that the settlement, commitment, and incentive-enforcement paths of such a credit-limit-based design can be realized with low on-chain overhead, without modifying the underlying rollup protocol.

Our contributions are threefold. First, we characterize the incentive conditions required for decentralized non-custodial micropayment systems to support credit limits beyond full collateralization. Second, we identify the trade-off governing such credit-based operation: improving capital efficiency requires enforcement conditions that bound strategic default incentives under repeated buyer--merchant interaction. Third, we clarify a regime in which under-collateralized credit can be sustained without custodial control, while making explicit the assumptions on exposure, verifiable settlement outcomes, and future participation value under which this support holds.

The remainder of this paper is organized as follows.
Section II summarizes the application-layer instantiation used to ground the incentive model. Section~\ref{sect:micropayment_formal_model} formalizes the game-theoretic model.
Sections~\ref{sect:micropayment_merchant_incentives} and
\ref{sect:micropayment_buyer_incentives} analyze incentive and equilibrium properties for
merchants and buyers.
Section~\ref{sect:micropayment_implementation} discusses implementation feasibility.
Section~\ref{sect:micropayment_experimental_evaluation} presents experimental results.
Section~\ref{sect:micropayment_discussion_limitations} discusses scope and limitations.
Section~\ref{sect:micropayment_related_work} reviews related work.
Section~\ref{sect:micropayment_conclusion} concludes.

\section{System and Mechanism Overview}
\label{sect:micropayment_overview}
This section summarizes an application-layer instantiation of the incentive conditions characterized in this paper. The instantiation combines (i) a stake-based credit limit mechanism with history-dependent access, (ii) an over-limit auction for temporary liquidity extension, and (iii) contract-enforced reward and penalty mechanisms that update access privileges based on publicly verifiable outcomes for buyers and merchants.

The design supports continued off-chain transaction processing through delayed settlement via batch processing, which reduces on-chain workload. Off-chain execution is modeled as untrusted and verifiable, with no reliance on honest execution; all incentive-relevant outcomes are enforced through on-chain verification of committed state. The term ``reputation'' is used in this paper only as a conceptual description: 
the protocol does not implement a standalone reputation or trust-scoring module. Instead, reputation effects arise endogenously from publicly observable histories and penalty states, together with agents’ equilibrium continuation values used purely for analytical characterization.

Before introducing the protocol components, Table~\ref{tab:core_terms} summarizes the main terminology used throughout the paper. The table is not intended as a complete notation glossary; it is included to clarify the roles of the main economic, monitoring, and implementation terms before the formal model is introduced.

\begin{table*}[t]
\centering
\caption{Core terminology used in the incentive characterization.}
\label{tab:core_terms}
\begin{tabular}{p{0.18\textwidth} p{0.16\textwidth} p{0.58\textwidth}}
\toprule
Term & Symbol or notation & Role in this paper \\
\midrule
Buyer & \(i\in B\) & Receives goods or services and repays unsettled obligations under a credit limit. \\

Merchant & \(j\in M\) & Accepts micropayments and fulfills service or delivery obligations. \\

Credit limit & \(CL_i\) & Maximum unsettled value that buyer \(i\) can use before settlement. \\

Trust state & \(T_{i,t}\) & Public state summarizing incentive-relevant history; it is not a separate reputation module. \\

Epoch & \(t\) & Discrete settlement period; rewards, penalties, and credit updates are evaluated at epoch boundaries. \\

Public signal & \(y_t\) & Verifiable settlement signal observed by all participants after epoch \(t\). \\

Public history & \(y^t=(y_0,\ldots,y_{t-1})\) & Sequence of public signals observed before epoch \(t\); public strategies condition on this history. \\

Continuation value & \(V_{i,\mathrm{cont}}\) & Discounted value of future credit access and protocol-mediated transactions. \\

Over-limit auction & -- & Bounded epoch-level mechanism for temporary credit extension beyond the current credit limit. \\

\bottomrule
\end{tabular}
\end{table*}

\paragraph{Credit limits and off-chain accounting}
Each buyer is admitted to transact under an individual credit limit determined by financial stake and behavioral trust, allowing micropayments to be authorized without immediate on-chain settlement. Credit usage is tracked through a local off-chain usage record and an offline record of pending transactions, maintained solely for execution convenience rather than as authoritative state, with all incentive-relevant outcomes ultimately validated against on-chain commitments and publicly verifiable proofs at settlement. Instead of submitting each micropayment on-chain, transactions are aggregated into a Merkle tree, and a single root is committed on-chain at the settlement boundary, preserving auditability for settlement and post hoc verification. The protocol monitors each buyer’s cumulative unsettled value relative to the assigned credit limit, raises a misuse flag when outstanding usage exceeds the permitted bound, and routes requests that exceed the credit limit but remain within admissible risk bounds to the over-limit auction mechanism.

\paragraph{Operational Flow}
Micropayment transactions are processed off-chain within each settlement epoch. When a buyer initiates a payment, the protocol immediately checks and consumes the buyer’s available credit limit and records the transaction in an off-chain usage log. Individual transactions are not submitted on-chain in real time; instead, transactions accumulated during the epoch are aggregated into a Merkle tree, and a single root is committed on-chain at the settlement boundary. During settlement, individual transactions and their associated credit usage can be verified by submitting standard Merkle inclusion proofs against the committed root, allowing the smart contract to deterministically enforce rewards, penalties, and state transitions without storing per-transaction data on-chain.

\paragraph{Over-limit auction for temporary credit extension}
When a payment request exceeds the available credit limit by a marginal liquidity deficit, the buyer may request an over-limit auction to obtain a temporary credit extension under controlled risk. 
The auction is executed as a stage-wise, epoch-isolated mechanism: while buyers may trigger over-limit auctions in multiple epochs, at most one such auction is permitted per buyer in each epoch, and all auction outcomes are settled within the same settlement period. 
This design preserves atomicity by restricting each buyer to at most one over-limit auction per epoch and requiring all auction outcomes to be settled within the same settlement period, thereby bounding short-term credit exposure across epochs. 
The strategic implications of credit expansion are addressed in the repeated-game analysis.

Guarantors participate through a standard two-phase commit--reveal auction process. 
In the commit phase, eligible guarantors submit cryptographic commitments to their bids without revealing bid values, together with stake locked in the smart contract, which is released or slashed upon auction settlement. 
In the subsequent reveal phase, bids are disclosed and verified against prior commitments, and only valid reveals are considered for winner determination. 
The auction adopts a reverse Vickrey design \cite{vickrey1961counterspeculation}. 
When an over-limit auction is triggered, the buyer specifies a protocol-enforced upper bound on the acceptable borrowing cost; bids exceeding this bound are rejected, and the auction fails if no admissible bid exists. 
Upon auction settlement, the winning guarantor directly transfers the corresponding payment amount to the merchant, allowing the transaction to complete without routing funds through the buyer. 
Public on-chain signals, including trust scores, are used to align risk information across participants and mitigate common-value effects, while guarantors’ private costs determine allocation outcomes.

\paragraph{Incentive mechanisms, enforcement, and time structure}
Incentives are enforced through deterministic, contract-executed rewards and penalties, based on verifiable outcomes at settlement boundaries. On-chain actions,\footnote{%
Off-chain execution may affect transaction liveness, including censorship.
We model off-chain execution as untrusted but verifiable:
only transactions with valid on-chain evidence affect settlement or incentives.} including timely repayment, are directly verifiable through transaction timestamps and value transfers. Off-chain delivery outcomes are abstracted as externally verifiable public signals, without committing to a specific realization mechanism. The reward mechanism provides credit-related benefits and stake-based compensation. The penalty mechanism imposes monetary sanctions, including late payment penalties and collateral slashing upon default, as well as privilege reductions, through trust score degradation and credit contraction. Beyond immediate penalties, the protocol applies a finite-state control process, consisting of a punishment phase of fixed duration followed by a recovery phase with rate-limited trust regeneration. The system operates in discrete epochs of four hours, with a deterministic settlement timeline, aligning reward and penalty accounting with the repeated-game analysis developed later. On-chain actions are verifiable through transaction timestamps and value transfers, while off-chain delivery outcomes correspond to externally verifiable public signals $y_t$. These verifiable outcomes define the public signals used in the incentive analyses
that follow. 

This paper therefore studies the verifiable-delivery regime for merchant-side incentive enforcement. In this regime, fulfillment outcomes are admissible for the analysis only when service completion, delay, or non-delivery can be represented by a digitally verifiable or publicly auditable signal before settlement. The model decouples event detection from economic enforcement: the settlement contract does not need to observe physical delivery directly, but it must receive a reliable public signal \(y_t\) that can trigger rewards, penalties, and trust-state updates. Such a signal may be produced by a cryptographic delivery acknowledgment, a logistics-oracle proof, a digitally logged service-completion event, or a dispute-resolution layer whose output is accepted by the settlement contract. Physical-world sensing, oracle construction, and unverified buyer self-reporting are therefore outside the scope of this characterization. Extending the framework to those settings would change the monitoring structure and require a separate analysis of signal generation and dispute resolution.

\paragraph{Epoch-Level State Commitment and Execution Model}
\label{sect_overview_offchain_node}
We now make explicit the execution semantics and trust assumptions
associated with the off-chain aggregation and epoch-level state commitment
used throughout the protocol. The protocol processes micropayment transactions off-chain under a discrete epoch structure. Each transaction is assigned to exactly one settlement epoch at the time of submission. Incentive-relevant state transitions, including credit consumption, reward accrual, and penalty enforcement, are evaluated only at epoch boundaries. Transactions recorded within the same epoch are aggregated and summarized by a single Merkle root. For each epoch \(t\), at most one state commitment in the form of a Merkle root \(R_t\) is accepted by the settlement contract, indexed by the epoch identifier. Root submission is regulated by the settlement contract, which enforces a fixed submission window and rejects multiple submissions for the same epoch as well as submissions referencing past epochs. Committed roots therefore form a strictly ordered, append-only sequence \(\{R_0, R_1, \ldots\}\) enforced at the contract level. Once a root is recorded on-chain, it is immutable and cannot be replaced or rolled back, providing epoch-level durability for previously settled state.

The off-chain component responsible for aggregating transactions and constructing Merkle trees is treated as untrusted. Correctness does not rely on honest execution by this component. Attempts to submit stale roots, omit transactions, or equivocate across epochs do not permit retroactive modification of committed state and do not alter incentive outcomes associated with earlier epochs. Any transaction that affects settlement, incentive enforcement, or credit state must be accompanied by a valid Merkle inclusion proof against the committed root of its corresponding epoch. Inclusion proofs may be submitted by any party at settlement to trigger deterministic enforcement by the on-chain contract. Transactions not accompanied by valid inclusion proofs have no protocol effect. Omission of a transaction from a committed root therefore results only in the exclusion of that transaction and does not compromise previously committed state.

The protocol maintains the following invariants across epochs. First, \emph{epoch atomicity} ensures that incentive-relevant effects are evaluated only at settlement boundaries. Second, \emph{append-only commitment} ensures that committed roots are immutable and strictly ordered by epoch. Third, \emph{verifiable inclusion} ensures that no off-chain transaction can affect protocol state without an on-chain verifiable proof. These invariants rule out stale root submission, rollback of settled state,
duplicate commitment, and retroactive state modification, while keeping
on-chain overhead bounded.

\section{Formal Model}
\label{sect:micropayment_formal_model}
This section formalizes buyer–merchant interactions induced by the proposed non-custodial Layer~2 micropayment protocol. 
Buyers and merchants\footnote{As explained in Section~\ref{sect:micropayment_overview}, the over-limit auction follows a Vickrey stage-wise design under a renewal-exclusion rule, so guarantors’ incentives are resolved at the epoch level and excluded from the repeated-game formulation.} 
interact repeatedly, and the protocol defines protocol-level abstract state variables, including trust states, credit limits, fee multipliers, and penalty states. 
The evolution of these states is constrained by publicly observable and verifiable outcomes, but does not require full persistence of all state variables on-chain. 
As a result, deviations in one period affect future eligibility, rewards, and penalties through history-dependent state transitions.

We adopt PPE and apply the one-shot deviation
principle (OSDP) by comparing the payoff from a one-shot deviation with the
discounted continuation value under conforming behavior.

\subsection{Repeated-Game Framework}
\label{sec:repeated_game_framework}
We use a public monitoring structure aligned with the repeated-game literature~\cite{fudenberg1991game, mailath2006repeated}, where incentive-relevant events, including payments, delivery outcomes, and penalty triggers, are publicly observable through on-chain signals and cryptographic commitments, including Merkle root commitments and commit--reveal auction records, enforced by smart contracts. The analysis characterizes the conditions under which IC and liveness can be sustained under credit expansion, with formal parameter definitions introduced in subsequent sections.

We model the interaction as an infinitely repeated game with hidden protocol state and public monitoring~\cite{shoham2008multiagent}.\footnote{
We assume blockchain persistence so that confirmed transactions and incentive-relevant state transitions are publicly observable. Ordering uncertainty due to latency or fee dynamics is abstracted by modeling interaction at a discrete settlement granularity, formalized as epochs in subsequent sections.
}
Each agent $i$ has a publicly observable trust state $T_{i,t}$ and a credit-limit state $CL_{i,t}$ that is not explicitly maintained on-chain. The evolution of $CL_{i,t}$ is constrained by publicly verifiable outcomes and enforced indirectly through contract-executed penalties and access restrictions, without requiring the credit limit itself to be persistently stored on-chain. While the numerical value of $CL_{i,t}$ is not directly exposed, its effects are reflected in publicly observable outcomes, including transaction acceptance, fee adjustments, and trust state updates. Because the credit-limit state $CL_{i,t}$ is not fully observable, standard subgames defined over full action histories are not directly applicable; instead, we restrict attention to public strategies that condition only on publicly observable signal histories, as in PPE.

\paragraph{Standing Assumptions}
We consider an infinitely repeated game $G^\infty$ under deterministic public monitoring. Action sets $(A_k)_{k\in N}$ are finite, payoffs are bounded, and all players discount future utilities with a common discount factor $\delta \in (0,1)$. Public histories are denoted by $y^t = (y_0,\dots,y_{t-1})$, and the public state $s_t$ summarizes on-chain information relevant for stage payoffs and continuation utilities. All incentive-relevant signals, including payment outcomes, delivery status, incentive triggers, and trust state updates, are emitted on-chain at the end of each period and are observable by all participants. We restrict attention to pure strategies.

\subsection{Players, Actions, Public State, and Public Signal}
Let $N = \mathcal{B} \cup \mathcal{M}$ denote the player set, consisting of buyers $i \in \mathcal{B}$ and merchants $j \in \mathcal{M}$. In each period $t$, buyers choose
\[
A_i = \{\text{PayOnTime},\ \text{LatePay},\ \text{Default}\},
\]
and merchants choose
\[
A_j = \{\text{DeliverOnTime},\ \text{LateDeliver},\ \text{FailToDeliver}\}.
\]
The joint action profile is $a_t \in A = \prod_{k\in N} A_k$.

The public state $s_t$ summarizes on-chain information used for stage payoffs and continuation utilities, including trust scores and aggregate system balances. After execution of $a_t$, the settlement layer emits a public signal $y_t \in \mathcal{Y}$ deterministically. The signal summarizes verifiable outcomes, including payment confirmations, delivery outcomes, incentive triggers, and trust score updates. Since $s_t$ and $y_t$ evolve deterministically under smart contract rules, players reconstruct public histories from observed signals.

\subsection{Solution Concept and Incentive Condition}
We characterize equilibrium behavior under PPE~\cite{fudenberg1994folk}, restricting attention to public pure strategies that condition on public signal histories.
Off-chain actions are not directly observable, but incentive-relevant outcomes are publicly verifiable at epoch boundaries via on-chain settlement signals.
Under this public monitoring structure, equilibrium verification follows the OSDP.
The deviation gain and the discounted continuation loss, arising from protocol-enforced penalties and the loss of future participation, are instantiated using the protocol’s reward, penalty, and cost components. Satisfaction of the resulting OSDP inequality is sufficient for sustaining PPE under public monitoring~\cite{mailath2006repeated}. 

In this work, time is indexed by periods $t=0,1,\dots$, where each period corresponds to a protocol epoch, $\tau_{\text{epoch}} = 4$ hours. Micropayments may occur off-chain within an epoch, while settlement and incentive distribution occur atomically at the epoch boundary. A participant is conforming in period $t$ only if all obligations accrued within that epoch are fulfilled. Stage payoffs in subsequent sections represent epoch-level aggregate utility.

\section{Merchant Incentive and Equilibrium Analysis}
\label{sect:micropayment_merchant_incentives}
This section analyzes merchant participation incentives and characterizes system liveness on the supply side. Using the repeated-game framework defined in Section~\ref{sec:repeated_game_framework}, we derive the IC and individual rationality (IR) constraints under which rational merchants maintain availability and comply with protocol deadlines. In high-frequency micropayment settings, such as a point-of-sale payment node, merchant $j$ acts as a payment service provider responsible for execution reliability and liquidity management. Execution reliability requires that accepted transactions are completed within protocol deadlines, while liquidity management concerns capital allocation under opportunity cost constraints. The analysis identifies IC and IR conditions under which merchant availability is sustained, assuming functional infrastructure.

\subsubsection{Stage-Game Utility Functions}
\label{sec:merchant_utility_functions}

This subsection defines the utility functions associated with the merchant action set $A_j$. For analysis, the actions $\textit{DeliverOnTime}$, $\textit{LateDeliver}$, and $\textit{FailToDeliver}$ are denoted by $a_{conf}$, $a_{late}$, and $a_{def}$.

\begin{itemize}
    \item $a_{conf}$ (\textbf{Conform}): The merchant completes delivery within the protocol-defined deadline $\Delta t$.
    \item $a_{late}$ (\textbf{Delay}): The merchant completes delivery after the deadline, retains control over transaction-related resources during the delay, and incurs a late-delivery penalty $P_j^{LM}$.
    \item $a_{def}$ (\textbf{Default}): The merchant fails to deliver and exits the protocol, retaining transaction value while triggering confiscation of staked collateral.
\end{itemize}

Let $\mathcal{T}$ denote the set of transactions processed in the current settlement epoch, indexed by $k$, with payment amount $v_k$. For merchant $j$, let $u_j(a)$ denote the aggregate utility induced by action $a$, defined as the sum of payoffs over all transactions in $\mathcal{T}$. Utilities are evaluated at the epoch boundary, consistent with the protocol’s epoch-based settlement model.\footnote{Although stake release is subject to a delay, incentive constraints are evaluated per epoch.
Delays and defaults are detected and penalized at epoch boundaries, which bounds deviation payoffs by single-epoch transaction volume.
Accordingly, the required stake $S_j^M$ covers single-epoch exposure.} 

For $a_{late}$, incentives are evaluated under the worst-case deviation in which all transactions in $\mathcal{T}$ are delayed. The epoch-level utility function is
\begin{equation}
\label{eq:micropayment_merchant_payoffs_final}
u_j(a) =
\begin{cases}
\begin{aligned}
    & \sum_{k \in \mathcal{T}} (v_k - C_{j,k}^{fee} - C_{j,k}^{exec} \\
    & \quad + R_{j,k}^{FM}) + R_j^{SM} - C_j^{stake},
\end{aligned}
& \text{if } a = a_{conf}; \\[20pt]

\begin{aligned}
    & \sum_{k \in \mathcal{T}} \big( v_k - C_{j,k}^{fee} - C_{j,k}^{exec} \\
    & \quad + R_{j,k}^{FM} + \Psi_{j,k} - P_{j,k}^{LM} \big) \\
    & \quad + R_j^{SM} - C_j^{stake},
\end{aligned}
& \text{if } a = a_{late}; \\[20pt]

\begin{aligned}
    & \sum_{k \in \mathcal{T}} (v_k + \Psi_{j,k} - P_{j,k}^{DM}) \\
    & \quad - C_j^{stake},
\end{aligned}
& \text{if } a = a_{def}.
\end{cases}
\end{equation}

$C_{j,k}^{fee}$ and $C_{j,k}^{exec}$ denote the per-transaction protocol fee and execution cost incurred by merchant $j$, respectively.
$R_{j,k}^{FM}$ is the transaction-level fee rebate granted upon successful delivery, derived from the aggregate epoch-level rebate $R_j^{FM}$ allocated to merchant $j$.
$R_j^{SM}$ denotes the fixed per-epoch reward associated with maintaining a compliant staking position.
$C_j^{stake}$ represents the opportunity cost of capital locked as stake over the epoch, modeled as the foregone return on staked capital.
$P_{j,k}^{LM}$ and $P_{j,k}^{DM}$ denote the protocol-enforced penalties for late delivery and default, respectively.
$\Psi_{j,k}$ denotes an upper bound on the external opportunity payoff obtainable from temporarily retaining control over transaction-related resources associated with transaction $k$, such as liquidity float.

\subsubsection{Incentive Compatibility Analysis}
\label{sec:ic_analysis}
We characterize the merchant’s IC constraints within a single settlement epoch. The analysis identifies conditions on $P_{j,k}^{LM}$ and the merchant stake $S_j^M$ under which no profitable deviation exists, even with heterogeneous external opportunity gains.

Two deviation classes are considered, delay $a_{late}$ and default $a_{def}$.
For each, we derive conditions enforcing the strict preference ordering
$a_{conf} \succ_j a_{late} \succ_j a_{def}$, where $\succ_j$ denotes strict preference
with respect to merchant $j$’s epoch-level utility.
Preventing $a_{late}$ ensures timely settlement, while preventing $a_{def}$ ensures that default does not become the best response.

\begin{proposition}[Dominance of Timely Delivery]
\label{prop:micropayment_dominance_of_timely_delivery}
If $P_{j,k}^{LM} > \Psi_{j,k}$ for all $k \in \mathcal{T}$, then
$a_{conf} \succ_j a_{late}$.
\end{proposition}

\begin{proof}
From Eq.~\ref{eq:micropayment_merchant_payoffs_final},
\begin{equation}
u_j(a_{conf}) - u_j(a_{late})
= \sum_{k \in \mathcal{T}} \left( P_{j,k}^{LM} - \Psi_{j,k} \right).
\end{equation}
If $P_{j,k}^{LM} > \Psi_{j,k}$ for all $k \in \mathcal{T}$, then
$u_j(a_{conf}) > u_j(a_{late})$, and hence $a_{conf} \succ_j a_{late}$.
\end{proof}

To preclude default ($a_{def}$) from becoming the best response, we impose the following assumptions.

\begin{assumption}[Bounded Aggregate Liability]
\label{ass:micropayment_bounded_execution}
For any merchant $j$ and epoch transaction set $\mathcal{T}$,
\begin{equation}
\sum_{k \in \mathcal{T}} C_{j,k}^{exec} \le C_{\max}.
\end{equation}
\end{assumption}

\begin{assumption}[Persistent Alternative Utility]
\label{ass:micropayment_persistent_opportunity}
The external opportunity gain $\Psi_{j,k}$ is modeled as an exogenous upper bound
on deviation payoffs and is not assumed to increase through the choice of deviation strategy.
\end{assumption}

\begin{theorem}[Settlement Liveness with Bounded Liability]
\label{thm:micropayment_liveness_bounded_execution}
Fix merchant $j$ satisfying Assumptions~\ref{ass:micropayment_bounded_execution}
and~\ref{ass:micropayment_persistent_opportunity}. Let $\mathcal{T}$ denote the set
of transactions in the current epoch. If all transactions in $\mathcal{T}$ have
passed their execution deadlines, then $a_{late} \succ_j a_{def}$ whenever
\begin{equation}
R_j^{SM} + \sum_{k \in \mathcal{T}} P_{j,k}^{DM}
>
C_{\max} + \sum_{k \in \mathcal{T}}
\left( C_{j,k}^{fee} + P_{j,k}^{LM} - R_{j,k}^{FM} \right).
\end{equation}
\end{theorem}

\begin{proof}
Comparing $u_j(a_{late})$ and $u_j(a_{def})$ using
Eq.~\ref{eq:micropayment_merchant_payoffs_final}, common terms cancel and
rearrangement yields
\[
R_j^{SM} + \sum_{k \in \mathcal{T}} P_{j,k}^{DM}
>
\sum_{k \in \mathcal{T}}
\left( C_{j,k}^{exec} + C_{j,k}^{fee} + P_{j,k}^{LM} - R_{j,k}^{FM} \right).
\]
By Assumption~\ref{ass:micropayment_bounded_execution},
$\sum_{k \in \mathcal{T}} C_{j,k}^{exec} \le C_{\max}$, which implies the stated
sufficient condition.
\end{proof}

\begin{corollary}[Strict Dominance of Conforming]
\label{cor:micropayment_strict_dominance_conforming}
Proposition~\ref{prop:micropayment_dominance_of_timely_delivery}
and Theorem~\ref{thm:micropayment_liveness_bounded_execution}
imply the strict dominance ordering
\begin{equation}
u_j(a_{conf}) > u_j(a_{late}) > u_j(a_{def}).
\end{equation}
\end{corollary}

\begin{proof}
Proposition~\ref{prop:micropayment_dominance_of_timely_delivery} establishes
$u_j(a_{conf}) > u_j(a_{late})$, while
Theorem~\ref{thm:micropayment_liveness_bounded_execution} establishes
$u_j(a_{late}) > u_j(a_{def})$.
The strict ordering follows by transitivity; see Appendix~\ref{app:micropayment_incentive_proofs}.
\end{proof}

\subsubsection{Individual Rationality Analysis}
We characterize the IR condition required for merchant participation. IR requires that the expected utility from conforming behavior weakly dominates non-participation. Normalizing the non-participation option to zero yields the following necessary condition.

\begin{proposition}[Stake-Dependent Participation Constraint]
\label{prop:micropayment_participation_constraint}
A necessary condition for a merchant $j$ to satisfy IR is that
\begin{equation}
\label{eq:micropayment_ir_constraint_of_merchant}
\sum_{k \in \mathcal{T}} R_{j,k}^{FM} + R_j^{SM}
\;\ge\;
\sum_{k \in \mathcal{T}} \left( C_{j,k}^{exec} + C_{j,k}^{fee} \right)
+ C_j^{stake}.
\end{equation}
\end{proposition}

\begin{proof}
IR requires that the realized utility from conforming behavior
weakly dominates non-participation. From
Equation~\ref{eq:micropayment_merchant_payoffs_final},
\[
u_j(a_{conf}) =
\sum_{k \in \mathcal{T}} \left( R_{j,k}^{FM}
- C_{j,k}^{exec} - C_{j,k}^{fee} \right)
+ R_j^{SM} - C_j^{stake}.
\]
Rearranging terms yields the stated condition.
\end{proof}

\subsubsection{Long-Run Equilibrium Analysis}
We extend the analysis to an infinitely repeated interaction under public monitoring.
The repeated-game abstraction omits post-deviation dynamics, yielding a conservative estimate of the loss in continuation value, as restarting the interaction forfeits accumulated reputation and future incentive benefits. 

\noindent\textbf{Utility under Punishment.}
During punishment, protocol-defined incentive rewards are suspended.
Accordingly, the punishment utility corresponds to the conforming utility
$u_j(a_{conf})$ with the reward components removed:
\begin{equation}
u_j^{\mathcal{P}}(a_{conf}, s)
=
\sum_{k \in \mathcal{T}(s)} \left( v_k - C_{j,k}^{fee} - C_{j,k}^{exec} \right)
- C_j^{stake}.
\end{equation}

\noindent\textbf{Discounted Loss under Incentive Suspension.}
We characterize the discounted utility loss induced by incentive suspension over a
finite punishment horizon.

\begin{lemma}[Suspension Loss Bound]
\label{lem:micropayment_future_loss}
For a punishment duration $T$ and discount factor $\delta \in [0,1)$,
the discounted continuation loss satisfies
\begin{equation}
\Delta u_j^{\mathrm{loss}}
\;\ge\;
\left(\frac{1-\delta^T}{1-\delta}\right)\underline{\ell}_j .
\end{equation}
\end{lemma}

\noindent
The constant $\underline{\ell}_j > 0$ denotes a lower bound on the per-epoch utility loss
incurred by merchant $j$ due to the suspension of protocol-defined incentive rewards.

\begin{proof}
The bound follows from substituting the per-epoch loss into the discounted sum
and evaluating a finite geometric series.
A detailed derivation is provided in Appendix~\ref{app:micropayment_incentive_proofs}.
\end{proof}

This strict ordering over $\{a_{conf}, a_{late}, a_{def}\}$ implies that any one-shot
deviation is strictly dominated in the stage game. We obtain the following result.

\begin{theorem}[Perfect Public Equilibrium]
\label{thm:micropayment_merchant_ppe}
Under the strict stage-game ordering in
Corollary~\ref{cor:micropayment_strict_dominance_conforming}
and the continuation loss bound in
Lemma~\ref{lem:micropayment_future_loss},
the public strategy profile that prescribes $a_{conf}$
on the conforming path constitutes a PPE
for all $\delta \in [0,1)$.
\end{theorem}

\begin{proof}
The result follows from the One-Shot Deviation Principle
under public monitoring, combined with the stage-game ordering
and the continuation loss bound.
A complete proof is provided in Appendix~\ref{app:micropayment_incentive_proofs}.
\end{proof}

\subsubsection{Summary of Merchant Incentive Analysis}
\label{sec:merchant_incentive_summary}
Under the stated late-penalty condition, timely delivery strictly dominates delay. With bounded liability and collateral, delay strictly dominates default. Consequently, conforming behavior is the unique strictly dominant action in the stage game. This dominance extends to the infinite-horizon setting under public monitoring,
yielding a PPE for all discount factors.

\section{Buyer Incentive and Equilibrium Analysis}
\label{sect:micropayment_buyer_incentives}
This section analyzes buyer behavior under the protocol’s credit allocation rules.
The analysis differs from the merchant case in Section~\ref{sect:micropayment_merchant_incentives}.
While merchant compliance is enforced by strict stage-game dominance that holds for all discount factors,
the buyer protocol operates in an under-collateralized regime designed for capital efficiency.
To support high-frequency transactions, the system permits credit expansion, allowing the buyer’s
outstanding payment obligation $v_{i,t}$ to exceed the posted collateral.

In this regime, timely repayment is not a stage-game dominant strategy.
The one-shot gain from default can exceed the recoverable collateral, so static penalties alone
are insufficient to deter deviation.
Buyer incentives therefore rely on the loss of future credit access and the associated
foregone trade utility, which disciplines current behavior. We model the interaction as an infinitely repeated game with public monitoring and apply the One-Shot Deviation Principle to characterize IC and equilibrium existence
under credit expansion.

\subsubsection{Stage-Game Utility Functions}
\label{sec:buyer_stage_game_utility}

We define the buyer’s stage-game utility over a settlement epoch.
The actions \textit{PayOnTime}, \textit{LatePay}, and \textit{Default} are denoted by
$a_{conf}$, $a_{late}$, and $a_{def}$.

\begin{itemize}
    \item $a_{conf}$ (\textbf{Conform}): The buyer settles the full payment obligation within the settlement window.
    \item $a_{late}$ (\textbf{Delay}): The buyer delays payment within the cure period, retaining control
    over the payment principal while incurring late-payment penalties.
    \item \(a_{def}\) (Default): The buyer fails to repay by the end of the cure period, resulting in collateral confiscation and loss of future credit access under Assumption~\ref{assump:micropayment_buyer_grim_trigger}.
\end{itemize}

Let $\mathcal{T}$ denote the set of transactions settled in the current epoch, indexed by $k$.
For buyer $i$, let $w_k$ denote the service value and $v_k$ the payment amount of transaction $k$.
The term $\omega$ denotes the buyer’s monetary conversion factor that maps changes in
credit capacity into monetary-equivalent utility.
Accordingly, credit-limit rewards and penalties enter the stage utility through
multiplication by~$\omega$.
The buyer’s utility $u_i(a)$ is given by:

\begin{equation}
\label{eq:buyer_stage_game_utility}
u_i(a) =
\begin{cases}
\begin{aligned}
    & \sum_{k \in \mathcal{T}} (w_k - v_k) + R_i^{TR} \\
    & \quad + R_i^{SB} - C_i^{stake} + \omega \cdot R_i^{CB},
\end{aligned}
& \text{if } a = a_{conf}; \\[20pt]

\begin{aligned}
    & \sum_{k \in \mathcal{T}} \big( w_k - v_k + \Psi_{i,k} - P_{i,k}^{LB} \big) \\
    & \quad + R_i^{SB} - C_i^{fin} - C_i^{stake} \\
    & \quad - \omega \cdot P_i^{CB},
\end{aligned}
& \text{if } a = a_{late}; \\[20pt]

\begin{aligned}
    & \sum_{k \in \mathcal{T}} (w_k + \Psi_{i,k}) - P_i^{DB} \\
    & \quad - C_i^{stake} - \omega \cdot CL_i,
\end{aligned}
& \text{if } a = a_{def}.
\end{cases}
\end{equation}

\noindent
$\Psi_{i,k}$ denotes an upper bound on the buyer’s short-term benefit from retaining
control over the payment principal $v_k$ during the delay period.
$P_{i,k}^{LB}$ denotes the late-payment penalty associated with transaction $k$, while
$P_i^{DB}=S_i^B$ denotes the default penalty given by full confiscation of the buyer’s
staked collateral $S_i^B$.
$C_i^{\text{stake}}$ represents the opportunity cost of staking, and $C_i^{\text{fin}}$
denotes the financing cost incurred under delayed settlement.
$R_i^{TR}$ denotes a protocol-defined transaction-rebate reward, granted upon timely repayment and included as an auxiliary incentive instrument, and $R_i^{SB}$ denotes the stake-interest reward
earned within the epoch.
$R_i^{CB}$ denotes the credit-limit reward given by the increase in credit capacity under
timely settlement, and $P_i^{CB}$ the credit-limit penalty given by the contraction in
credit capacity under deviation.
$CL_i$ denotes the buyer’s credit limit at the beginning of the epoch.

\subsubsection{Stage-Game Incentive Compatibility}
\label{sec:buyer_stage_game_ic}

We first analyze incentives to deter strategic delay.
The following result establishes a sufficient condition under which timely repayment
strictly dominates delayed repayment within a single epoch.

\begin{proposition}[Dominance of Timely Repayment]
\label{prop:micropayment_timely_settlement_buyer}
For any buyer $i$, if the late-payment penalty satisfies
$P_{i,k}^{LB} > \Psi_{i,k}$ for all $k \in \mathcal{T}$,
then $a_{conf}$ strictly dominates $a_{late}$ at the stage-game level.
\end{proposition}

\begin{proof}
From Eq.~\ref{eq:buyer_stage_game_utility}, the utility gap
$u_i(a_{conf}) - u_i(a_{late})$ cancels common terms and reduces to
$\sum_{k \in \mathcal{T}} (P_{i,k}^{LB} - \Psi_{i,k})$ plus weakly non-negative terms,
including $R_i^{TR}$, $C_i^{\text{fin}}$, and $\omega ( R_i^{CB} + P_i^{CB} )$.
Under $P_{i,k}^{LB} > \Psi_{i,k}$ for all $k \in \mathcal{T}$, this gap is strictly positive,
so $a_{conf} \succ a_{late}$ in the stage game.
A complete proof is provided in Appendix~\ref{app:buyer_proofs}.
\end{proof}

This condition eliminates strategic delay but does not preclude default under credit expansion.
We therefore turn to repeated-game incentives.

\subsubsection{Repeated-Game Incentives and Default Deterrence}

We impose the following protocol-level assumptions governing settlement timing,
credit exposure, and default handling.

\begin{assumption}[Explicit Settlement Timeline]
\label{assump:buyer_timeline}
Repayment obligations are strictly bound by the epoch boundary $t_{due}$ and a fixed cure
period $\Delta T_{cure}$, beyond which any outstanding obligation is treated as default.
\end{assumption}

\begin{assumption}[Bounded Aggregate Exposure]
\label{assump:micropayment_buyer_bounded_credit}
For any buyer $i$ and any epoch, the aggregate outstanding debt,
defined as the sum of unpaid transaction values within the epoch,
satisfies $v_{i,t} \le v_{\max} < \infty$.
\end{assumption}

\begin{assumption}[Algorithmic Default Enforcement]
\label{assump:micropayment_buyer_grim_trigger}
If buyer \(i\) fails to repay \(v_{i,t}\) by \(t_{\mathrm{due}}+\Delta T_{\mathrm{cure}}\), the protocol confiscates collateral \(S_i^B\) and maps the default event into a loss of future credit access. This loss may be implemented through account suspension, credential revocation, restitution requirements, trust reset, or conservative credit rebuilding for re-entering identities. In the analytical model, this continuation-value loss is represented by setting the buyer's post-default credit continuation value to zero. Thus, the assumption characterizes a frictional re-entry regime rather than requiring that permissionless blockchain addresses be physically permanent.
\end{assumption}

We now characterize equilibrium existence under credit expansion.

\begin{theorem}[Conforming PPE with Credit Expansion]
\label{thm:micropayment_conforming_ppe_credit_buyer}
Consider the infinitely repeated game with discount factor $\delta \in (0,1)$
and bounded credit exposure $v_{i,t} \le v_{\max}$, which bounds the buyer’s
maximum one-shot deviation payoff.
There exists a threshold
\begin{equation}
\underline{\delta}
=
\frac{v_{\max} - S_i^B}{v_{\max} - S_i^B + \bar{u}_i}
\end{equation}
such that for all $\delta \ge \underline{\delta}$,
a PPE exists in which the buyer conforms in every epoch.
\end{theorem}

\begin{proof}
The equilibrium construction and incentive verification follow from the OSDP. A complete proof is provided in Appendix~\ref{app:buyer_proofs}.
\end{proof}

Theorem 3 should therefore be read as a characterization of the identity-friction regime under which buyer-side credit deterrence is meaningful. The permanent-exclusion language in the repeated-game construction represents the economic loss of future credit capacity, not a claim that a permissionless address cannot be recreated. The same incentive logic applies when a defaulting buyer cannot immediately regain the same credit capacity because re-entry requires a new identity credential, restitution, trust rebuilding, or a period of low initial credit limits. If identity reset were costless and immediately restored the same credit capacity, the continuation-value loss would collapse and the buyer-side incentive region would change. Such a setting is a different identity regime and requires a separate characterization of credit rebuilding costs and re-entry constraints.

Unlike the merchant case, buyer compliance under credit expansion is sustained by the
loss of future transaction opportunities rather than by full collateral coverage.

\begin{corollary}[Strategy Ordering with Recovery]
\label{cor:micropayment_strategy_ordering}
Under the equilibrium in Theorem~\ref{thm:micropayment_conforming_ppe_credit_buyer},
there exists a parameter regime such that
\begin{equation}
u_i(a_{conf}) > u_i(a_{late}) > u_i(a_{def}).
\end{equation}
\end{corollary}

\begin{proof}
The result follows from combining stage-game dominance of $a_{conf}$ over $a_{late}$
with the repeated-game deterrence of $a_{def}$.
A complete proof is provided in Appendix~\ref{app:buyer_proofs}.
\end{proof}

\subsubsection{Individual Rationality via Capital Efficiency}
We evaluate buyers' IR relative to non-participation with fully collateralized settlement.
Under credit expansion, buyers avoid locking capital equal to $v_{i,t} - S_i^B$, thereby reducing opportunity costs.

\begin{proposition}[Buyer Participation via Capital Efficiency]
\label{prop:microapyment_buyer_participation_constraint}
Participation is individually rational whenever
\begin{equation}
(v_{i,t} - S_i^B)\cdot r_{\mathrm{opp}} \cdot
\frac{\tau_{\mathrm{epoch}}}{365 \cdot 24}
\ge
C_i^{\mathrm{fin}} .
\end{equation}
\end{proposition}

\begin{proof}
Under non-participation with fully collateralized settlement, the buyer incurs an
epoch-level opportunity cost of
\[
(v_{i,t} - S_i^B)\cdot r_{\mathrm{opp}} \cdot \frac{\tau_{\mathrm{epoch}}}{365 \cdot 24},
\]
where $r_{\mathrm{opp}}$ is an annualized rate normalized using a standard 365-day
convention. Participation avoids this cost while incurring $C_i^{\mathrm{fin}}$.
IR therefore requires the avoided opportunity cost to weakly
exceed $C_i^{\mathrm{fin}}$, yielding the stated condition. See
Appendix~\ref{app:buyer_proofs} for details.
\end{proof}

\subsubsection{Summary of Buyer Incentive Analysis}
\label{sec:micropayment_buyers_incentive_summary}

The buyer-side analysis establishes IC under credit expansion.
Late-payment penalties eliminate strategic delay at the stage-game level.
Default is deterred in the repeated game by the loss of continuation value,
yielding a PPE above a critical discount factor. Graduated penalties admit a recovery path while preserving strict preference
for timely repayment. Participation is individually rational when the capital efficiency gains from credit expansion exceed financing costs.

\section{Implementation Feasibility}
\label{sect:micropayment_implementation}
This section uses an Arbitrum Nitro prototype to instantiate the incentive-critical execution and verification paths analyzed above. The purpose is not to present a full end-to-end payment system, but to show that the settlement, commitment, and incentive-enforcement paths required by the characterization can be realized at the application layer. User interfaces and identity management are treated as external deployment assumptions, as they do not affect the on-chain enforcement or equilibrium properties evaluated in this section. The prototype includes credit usage constraints, trust-state updates, incentive-triggered state transitions, over-limit auction logic, and deferred settlement with batched on-chain commitment. It is deployed on the Sepolia test network and does not modify the sequencing or dispute-resolution logic of the underlying Arbitrum rollup.

To reduce on-chain costs, micropayment transactions are accumulated off-chain within each
settlement epoch and organized into a Merkle tree, whose root commits to all transactions
in the epoch and is submitted on-chain. Transactions affecting settlement or incentives
are verified using standard Merkle inclusion proofs against the committed root, so
on-chain verification cost is constant per epoch and independent of batch size.
Relative to a baseline that submits transactions individually, this batching approach
reduces on-chain gas consumption by over ninety percent while preserving verifiability.

In sum, the prototype provides execution-level evidence that the incentive-critical paths analyzed in this paper can be realized on an existing Layer 2 rollup without custodial trust, protocol-level modification, or unbounded on-chain overhead.

\section{Experimental Evaluation}
\label{sect:micropayment_experimental_evaluation}

This section evaluates the on-chain cost characteristics of the illustrative instantiation under delayed settlement, using deployed smart contracts on the Arbitrum Nitro Sepolia test network. We compare direct on-chain submission without batching ($V_0$) against
batched settlement under credit limits with deferred commitment ($V_1$).

\subsection{Gas Efficiency under Delayed Settlement}

We compare direct on-chain submission ($V_0$) with batched settlement under credit limits
($V_1$) for transaction volumes ranging from 100 to 500 transfers per batch.
Table~\ref{tab:gas_comparison_usd} reports gas consumption under steady-state operation.
Across all batch sizes, batched settlement substantially reduces on-chain gas consumption
relative to direct submission.

\begin{table}[t]
  \centering
  \caption{Gas consumption under direct submission ($V_0$) and steady-state batched
  settlement ($V_1$). Percentage savings are computed relative to $V_0$.
  Steady-state costs exclude one-time initialization overhead.}
  \label{tab:gas_comparison_usd}
  \vspace{2pt}
  \setlength{\tabcolsep}{4pt}
  \begin{tabular}{c r r r}
    \toprule
    Tx Vol. & $V_0$ Gas & $V_1$ Gas & Save (\%) \\
    \midrule
    100 & 5,863,758 & 106,274 & 98.19 \\
    200 & 11,696,824 & 106,274 & 99.09 \\
    300 & 17,532,352 & 106,274 & 99.39 \\
    400 & 23,287,671 & 106,262 & 99.54 \\
    500 & 29,735,785 & 106,274 & 99.64 \\
    \bottomrule
  \end{tabular}
\end{table}

\subsection{Constant-Time On-Chain Verification}

To evaluate the scalability of on-chain verification, we measure the gas cost of committing
Merkle roots that summarize batched transaction states and residual credit states.
Table~\ref{tab:commitroot_gas_scala} shows that the gas cost of committing both transaction
and credit Merkle roots remains effectively constant across batch sizes, confirming that
on-chain verification overhead does not scale with transaction volume.

\begin{table}[t]
  \centering
  \caption{Gas consumption for committing transaction and credit Merkle roots under
  different transaction volumes. Each submission commits a single Merkle root and incurs
  constant on-chain cost.}
  \label{tab:commitroot_gas_scala}
  \vspace{2pt}
  \setlength{\tabcolsep}{3pt}
  \begin{tabular}{c r r}
    \toprule
    Tx Vol. & TxRoot Gas & CreditRoot Gas \\
    \midrule
    100 & 21,892 & 21,892 \\
    200 & 21,892 & 21,892 \\
    300 & 21,880 & 21,892 \\
    400 & 21,892 & 21,892 \\
    500 & 21,892 & 21,892 \\
    \bottomrule
  \end{tabular}
\end{table}

\noindent\textbf{Summary.}
These results provide implementation-level evidence that the incentive-critical settlement and verification paths can be realized on a Layer 2 rollup with bounded on-chain overhead.
\section{Discussion and Limitations}
\label{sect:micropayment_discussion_limitations}
This section clarifies the scope and assumptions under which the incentive analysis applies and identifies structural trade-offs inherent in the proposed setting.

\subsection{Equilibrium Scope and Deviation Assumptions}
The incentive analysis models buyers and merchants as rational agents interacting under
publicly observable signals, consistent with the repeated-game framework in
Section~\ref{sec:repeated_game_framework}. Deviations are treated as observable events that
trigger deterministic, contract-enforced penalties. The analysis focuses on
application-layer behavior and abstracts away from cryptographic failures, key compromise,
and consensus-layer attacks, which are orthogonal to incentive-based enforcement.

The buyer-side deterrence result should be interpreted under an identity-friction regime. The analysis does not require a globally permanent real-world identity, but it does require that default cannot be followed by costless re-entry with the same credit capacity. In practical deployments, this condition can be approximated through conservative initial credit limits, trust- and time-dependent credit expansion, restitution requirements, credential-based access, or rate-limited rebuilding after default. These mechanisms bound the return achievable per identity and preserve a nonzero continuation-value loss after default. If a system instead permits costless identity reset with immediate restoration of credit capacity, the buyer-side incentive region changes and the framework must be re-characterized under that weaker identity regime.

Guarantors are modeled as participants in isolated over-limit auctions rather than as
agents engaged in repeated delivery--payment interactions. Their incentives are governed
by collateral-backed, contract-level commitments, allowing guarantor behavior to be
analyzed at the epoch level without cross-epoch incentive coupling.

At the validator and consensus layers, collateral buffers serve cryptoeconomic security
objectives~\cite{durvasula2024robust}. In contrast, stake in our setting plays an
application-layer role, anchoring incentives rather than fully collateralizing transaction
value. Accordingly, the analysis focuses on IC and equilibrium
behavior along the protocol execution path.

\subsection{Collusion in Over-Limit Auctions}
The over-limit auction provisions temporary liquidity using a reverse Vickrey design
with commit--reveal bidding. While this structure enforces truthful bidding under
independent participation, collusion among guarantors cannot be fully ruled out in
permissionless settings. The protocol mitigates collusion risks through structural
constraints, including commit--reveal bidding and collateral-backed participation, but a
sufficiently large coalition controlling available liquidity could still coordinate to
raise effective interest rates. Designing auction mechanisms that remain robust to
collusion under anonymity is therefore outside the scope of this work.

\subsection{Privacy and Verifiability Trade-offs}
The protocol exposes a fundamental trade-off between privacy and public verifiability.
Sustaining incentive-compatible enforcement of auction outcomes requires on-chain
disclosure of winning bids, which enables trustless verification and consistent risk
assessment over repeated interactions. Intuitively, guarantors improve pricing only if
they can observe and learn from past outcomes; if bid and trust information is hidden,
participants must price against worst-case uncertainty. While cryptographic techniques
such as zero-knowledge proofs could conceal bid values, hiding this information obscures
shared default risk and prevents risk differentiation across epochs, degrading allocation
efficiency and penalizing low-risk participants over time. Accordingly, public data
availability in this setting is an economic requirement for incentive stability, rather
than a purely engineering choice.

\subsection{Operational and Information Constraints}
The incentive analysis assumes fully rational agents and focuses on instant, low-value
transactions with tightly coupled payment and delivery. Bounded rationality, non-instant
delivery, and participation frictions are outside the scope of this paper. In addition, batching transactions through off-chain aggregation prioritizes efficient per-epoch
verification and incentive enforcement, but does not provide cross-epoch transaction
traceability. Extensions that aim to preserve long-term auditability or enhance privacy
through cryptographic techniques such as zero-knowledge proofs involve additional trade-offs
among scalability, latency, and information disclosure and are not considered in this work.

\section{Related Work}
\label{sect:micropayment_related_work}

This section situates our work within prior research on non-custodial blockchain-based micropayments, incentive mechanisms for decentralized systems, and Layer~2 cost reduction, highlighting how our approach differs in terms of capital efficiency and incentive structure. 

\subsection{Non-Custodial Micropayment Systems}
Early work explored decentralized settlement mechanisms for currency exchange and service payments, demonstrating how smart contracts and on-chain deposits can replace centralized clearing and escrow services~\cite{wu2017exploration,kudva2020pebers}. These systems rely on correct participant behavior to preserve liveness and security, which motivates the design of incentive mechanisms beyond consensus-layer rewards.

Payment channels represent a prominent approach to scaling micropayments by moving transactions off-chain and committing only final states to the blockchain. Bitcoin’s Lightning Network~\cite{poon2016bitcoin} and Ethereum’s Raiden Network~\cite{raiden2025} exemplify this design. While payment channels improve throughput without modifying the underlying consensus protocol, they impose operational complexity. Participants must remain online or rely on watchtower services, external monitors that
observe the blockchain and react to fraudulent channel closures, to prevent fraud, which introduces additional trust assumptions and coordination overhead~\cite{zhang2021counter,hao2018fastpay,miller2019sprites}. Moreover, incomplete on-chain records complicate transaction traceability and dispute resolution. In contrast to channel-based designs, our work avoids persistent bilateral channels and instead relies on protocol-enforced incentives and credit limits to support high-frequency payments without requiring continuous online monitoring. State-channel systems such as Lightning and Raiden rely on fully pre-funded channels, where the maximum transferable value is bounded by funds locked at channel opening, rather than on stake-based penalties or credit expansion.

Several smart contract-based micropayment systems employ collateral to enforce correct behavior. FastPay~\cite{hao2018fastpay} allows buyers and sellers to deposit funds and introduces a broker role to temporarily raise credit limits, while enforcing penalties for deviations. Snappy~\cite{MavroudisWDKC20} requires both customers and merchants to lock collateral to guarantee transaction completion prior to blockchain finality. Other platforms rely on mutual collateral forfeiture to establish trust between trading parties~\cite{le2019implementation,hasan2018blockchain}. These approaches improve security but tightly couple transaction capacity to locked capital. Our work differs by explicitly targeting capital efficiency through dynamically adjusted credit limits and by formalizing IC under under-collateralized operation.

\subsection{Game-Theoretic Analysis of Blockchain Incentives}
Game theory has been widely applied to analyze incentives and strategic behavior in
blockchain systems, including consensus participation, mining strategies, and reward
allocation~\cite{liu2019survey,azouvi2021sok,min2019security,wang2023security}. Prior work
has employed a range of game-theoretic models to study decentralized environments, such
as mining pool selection, electricity trading, and resource allocation, often under
strategic uncertainty or incomplete information~\cite{zhang2023analysis,xia2020bayesian,
asheralieva2020bayesian,xu2017game}. Most existing analyses focus on consensus-layer security or infrastructure-level
incentives. By contrast, our work applies repeated-game analysis with public monitoring
to application-layer micropayment protocols, where IC depends on credit expansion, deferred settlement, and publicly observable execution outcomes. We
characterize equilibrium behavior of buyers and merchants under protocol-enforced
penalties, extending repeated-game incentive analysis beyond mining and consensus to
non-custodial payment mechanisms.

\subsection{On-Chain Cost Reduction in Layer~2 Systems}

Reducing on-chain transaction costs is a primary motivation for Layer~2 solutions. Payment channels reduce on-chain load by processing transactions off-chain~\cite{zhang2018anonymous,burchert2018scalable,galal2019efficient}, but their deployment in open retail environments is limited by the need to establish channels between individual users and merchants. Alternative approaches leverage sidechains or blockchain-of-blockchains architectures to lower fees~\cite{networkskale,kilpatrickalthea,kwon2019cosmos}. While these systems reduce costs, they introduce additional trust assumptions or operational complexity.

Non-custodial Layer~2 designs aim to retain the security guarantees of the Layer~1 chain while improving scalability. Commit chains such as NOCUST~\cite{khalil2018commit} and Plasma~\cite{poon2017plasma} aggregate off-chain activity and periodically commit state roots on-chain. These designs reduce transaction fees but face challenges related to finality guarantees and operator centralization. Rollup techniques further improve cost efficiency by batching transactions and submitting compressed data or proofs to the main chain~\cite{fekete2024trust,hoang2023sssm,liu2024prosecutor}. Our work complements rollup-based approaches by analyzing the incentive and equilibrium
conditions under credit-based execution, which reduces the frequency of on-chain
commitments by allowing deferred settlement. Transaction fees are incurred only when aggregated state is committed, allowing additional cost savings beyond batching alone.

Prior work has addressed scalability, security, and incentive alignment largely as separate concerns. This paper differs by characterizing the incentive conditions under which non-custodial execution and under-collateralized credit expansion can coexist in a repeated-game micropayment setting.

\section{Conclusion}
\label{sect:micropayment_conclusion}

This paper characterizes the incentive conditions under which non-custodial micropayments can support credit limits beyond full collateralization. Using a repeated-game framework with public monitoring, the analysis identifies how bounded exposure, verifiable settlement outcomes, and continuation value constrain strategic behavior under credit expansion.

For merchants, timely delivery is sustained through stage-game incentive conditions under bounded liability. For buyers, conforming repayment can be sustained in an under-collateralized regime when the loss of future credit access outweighs the one-shot gain from default. These results identify the IC conditions required for reward and penalty mechanisms to sustain conforming behavior under per-epoch interactions and long-run repeated play.

The Arbitrum Nitro prototype serves as an illustrative application-layer instantiation of the settlement, commitment, and incentive-enforcement paths used in the analysis. The implementation evidence shows that these paths can be realized with constant on-chain verification cost and substantial gas reduction, without modifying the underlying rollup protocol.

The results clarify the incentive and enforcement conditions under which non-custodial credit expansion can support capital-efficient micropayments beyond full collateralization.

\bibliographystyle{IEEEtran}
\bibliography{reference}

\appendix
\section{Proofs for Incentive Analysis}
\label{app:micropayment_incentive_proofs}

\subsection{Proofs for Merchant Incentive Analysis}
\label{app:merchant_proofs}

\subsubsection{Proof of Corollary~\ref{cor:micropayment_strict_dominance_conforming}}

By Proposition~\ref{prop:micropayment_dominance_of_timely_delivery},
the conforming action strictly dominates delay,
$u_j(a_{conf}) > u_j(a_{late})$.
By Theorem~\ref{thm:micropayment_liveness_bounded_execution},
late delivery strictly dominates abandonment,
$u_j(a_{late}) > u_j(a_{def})$.
The strict dominance ordering
$u_j(a_{conf}) > u_j(a_{late}) > u_j(a_{def})$
follows by transitivity.

\subsubsection{Proof of Lemma~\ref{lem:micropayment_future_loss}}

For any public state $s$ reachable during punishment, the difference between the
conforming utility $u_j(a_{conf}, s)$ and the punishment utility
$u_j^{\mathcal{P}}(a_{conf}, s)$ corresponds to the suspension of protocol-defined
incentive rewards:
\begin{equation}
u_j(a_{conf}, s) - u_j^{\mathcal{P}}(a_{conf}, s)
=
\sum_{k \in \mathcal{T}(s)} R_{j,k}^{FM} + R_j^{SM}.
\end{equation}

This difference admits a per-epoch lower bound $\underline{\ell}_j > 0$.
Summing this loss over a punishment duration of $T$ epochs with discount factor
$\delta$ yields
\begin{equation}
\Delta u_j^{\mathrm{loss}}
=
\sum_{\tau=1}^{T} \delta^{\tau-1} \underline{\ell}_j
=
\left(\frac{1-\delta^T}{1-\delta}\right)\underline{\ell}_j,
\end{equation}
which establishes the stated bound.

\subsubsection{Proof of Theorem~\ref{thm:micropayment_merchant_ppe}}

Fix an arbitrary public history and the induced public state.
By Corollary~\ref{cor:micropayment_strict_dominance_conforming},
the conforming action $a_{conf}$ strictly dominates all feasible
one-shot deviations in the stage game.
Hence, any deviation yields a strictly negative immediate gain
relative to $a_{conf}$.

Any deviation is publicly observed and triggers entry into the
punishment phase.
By Lemma~\ref{lem:micropayment_future_loss},
the resulting discounted continuation loss is non-negative
for all $\delta \in [0,1)$. Therefore, for every public history and every feasible one-shot deviation,
the one-shot deviation condition is satisfied.
By the OSDP for repeated games with public monitoring,
the prescribed public strategy profile constitutes a PPE for all discount factors $\delta \in [0,1)$.


\subsection{Proofs for Buyer Incentive Analysis}
\label{app:buyer_proofs}

\subsubsection{Proof of Proposition~\ref{prop:micropayment_timely_settlement_buyer}}
\label{app:proof_prop_timely_buyer}

Define the utility gap
$\Delta u_i = u_i(a_{conf}) - u_i(a_{late})$.
From Eq.~\ref{eq:buyer_stage_game_utility},
\begin{equation}
\label{eq:buyer_utility_gap_conf_late}
\Delta u_i
=
\sum_{k \in \mathcal{T}} \left( P_{i,k}^{LB} - \Psi_{i,k} \right)
+ R_i^{TR} + C_i^{fin}
+ \omega \left( R_i^{CB} + P_i^{CB} \right).
\end{equation}

All terms except the first summation are weakly non-negative:
$R_i^{TR} \ge 0$ by construction, $C_i^{fin} \ge 0$, and
$\omega \cdot ( R_i^{CB} + P_i^{CB} ) \ge 0$ under $\omega \ge 0$ and
non-negative credit-capacity changes induced by conforming versus delayed behavior.
Therefore, a sufficient condition for $\Delta u_i > 0$ is
\begin{equation}
\label{eq:buyer_sufficient_cond_penalty_level}
\sum_{k \in \mathcal{T}} P_{i,k}^{LB}
>
\sum_{k \in \mathcal{T}} \Psi_{i,k}.
\end{equation}

In particular, if $P_{i,k}^{LB} > \Psi_{i,k}$ for all $k \in \mathcal{T}$, then
$\sum_{k \in \mathcal{T}} ( P_{i,k}^{LB} - \Psi_{i,k} ) > 0$, implying $\Delta u_i > 0$.
Hence, $u_i(a_{conf}) > u_i(a_{late})$, and $a_{conf}$ strictly dominates $a_{late}$
at the stage-game level.

\subsubsection{Proof of Theorem~\ref{thm:micropayment_conforming_ppe_credit_buyer}}

We construct a public strategy profile $\sigma^*$ as follows.
The buyer plays \(a_{conf}\) if the public history is clean, meaning that no default has occurred, and remains on the default path otherwise. Merchants serve if the history is clean and refuse service after default.

We verify incentive compatibility using the One-Shot Deviation Principle.
On the equilibrium path, serving a conforming buyer yields non-negative
utility for merchants.
After default, serving yields no payment and positive execution cost,
so refusal is optimal. For the buyer, compare conforming to default at any epoch.
The continuation value on the equilibrium path is
\begin{equation}
V_{i,\mathrm{cont}}
=
\sum_{\tau=1}^{\infty} \delta^{\tau-1} \bar{u}_i
=
\frac{\bar{u}_i}{1-\delta}.
\end{equation}

From Eq.~\ref{eq:buyer_stage_game_utility},
the one-shot gain from default relative to conforming
is upper bounded by $v_{i,t} - S_i^B$.
The buyer incentive constraint is therefore
\begin{equation}
\delta \cdot V_{i,\mathrm{cont}} \ge v_{i,t} - S_i^B.
\end{equation}

Substituting the continuation value and
considering the worst-case exposure $v_{i,t} = v_{\max}$
yields
\begin{equation}
\frac{\delta \cdot \bar{u}_i}{1-\delta} \ge v_{\max} - S_i^B.
\end{equation}

Rearranging terms gives
\begin{equation}
\delta \cdot (\bar{u}_i + v_{\max} - S_i^B) \ge v_{\max} - S_i^B,
\end{equation}
which implies
\begin{equation}
\underline{\delta}
=
\frac{v_{\max} - S_i^B}{v_{\max} - S_i^B + \bar{u}_i}.
\end{equation}

For all $\delta \ge \underline{\delta}$,
the buyer incentive constraint is satisfied,
establishing the existence of a conforming PPE.

\subsubsection{Proof of Corollary~\ref{cor:micropayment_strategy_ordering}}

We establish the strict strategy ordering in two steps. 

First, by Theorem~\ref{thm:micropayment_conforming_ppe_credit_buyer} and Assumption~\ref{assump:micropayment_buyer_grim_trigger}, default is not a best response because the permanent loss of future credit access along the default path outweighs any one-shot gain. Under delayed repayment, the buyer retains future credit access
while incurring a finite late-payment penalty.
Hence $a_{late} \succ a_{def}$ holds for admissible penalty parameters. 

Second, by Proposition~\ref{prop:micropayment_timely_settlement_buyer},
there exist penalty parameters such that
$u_i(a_{conf}) > u_i(a_{late})$.
Combining the two inequalities yields
\begin{equation}
u_i(a_{conf}) > u_i(a_{late}) > u_i(a_{def}),
\end{equation}
which establishes the claim.

\subsubsection{Proof of Proposition~\ref{prop:microapyment_buyer_participation_constraint}}
We compare protocol participation with a fully collateralized outside option. Under the outside option, the buyer pre-funds the full payment amount $v_{i,t}$
and incurs the corresponding opportunity cost.
The utility is
\begin{equation}
U_{i,t}^{\mathrm{out}}
=
U_{i,t}^{\mathrm{all}}
- v_{i,t}
- v_{i,t} \cdot r_{\mathrm{opp}} \cdot
\frac{\tau_{\mathrm{epoch}}}{365 \cdot 24}.
\end{equation}

Under protocol participation, the buyer locks only $S_i^B \le v_{i,t}$ and
incurs the financing cost $C_i^{\mathrm{fin}}$.
The utility is
\begin{equation}
U_{i,t}^{\mathrm{in}}
=
U_{i,t}^{\mathrm{all}}
- v_{i,t}
- S_i^B \cdot r_{\mathrm{opp}} \cdot
\frac{\tau_{\mathrm{epoch}}}{365 \cdot 24}
- C_i^{\mathrm{fin}}.
\end{equation}

Participation is individually rational whenever
$U_{i,t}^{\mathrm{in}} \ge U_{i,t}^{\mathrm{out}}$.
Canceling common terms yields
\begin{equation}
(v_{i,t} - S_i^B) \cdot r_{\mathrm{opp}} \cdot
\frac{\tau_{\mathrm{epoch}}}{365 \cdot 24}
\ge
C_i^{\mathrm{fin}},
\end{equation}
which establishes the participation condition.

\end{document}